%% file: main.tex
\documentclass[a4paper, 11pt]{article}
\input{preamble.tex}

\title{The Paradox of the Third Particle is classical}
\date{}
\author{Ladina Hausmann and Renato Renner}
\affil{\small \it Institute for Theoretical Physics, ETH Zürich, 8093 Z\"urich, Switzerland}

\begin{document}
\maketitle
\vspace{-1.6cm}

\begin{abstract} 
The Paradox of the Third Particle arises when particles are described relative to one of them serving as a reference frame. Although it is typically attributed to quantum superpositions, we show that the paradox and its ramifications---such as the relativity of subsystems---already occur classically. In light of this insight, we establish a no-go theorem that holds whenever physical subsystems, whether classical or quantum, are used as reference frames: if one demands that frame transformations be information-preserving, then it is impossible to meaningfully partition the world into subsystems, like individual particles, from the perspective of each frame.
\end{abstract}

\section{Introduction}
By the relativity principle, the position of a particle $B$ is defined only relative to another system, such as a particle $A$. This relational description defines what we call  the \emph{perspective of $A$}. Assuming a fully quantum treatment, where particles may be in spatial superposition, $A$'s own quantum properties affect its perspective. This forces us to revisit a basic question: how does one transform the description of $B$ from $A$'s perspective to that of $A$ from $B$'s perspective? Formulating such quantum reference frame transformations has been extensively studied~\cite{Aharonov1967,Aharonov1967a,Page1983,Aharonov1984,Rovelli1991,Bartlett2007,Palmer2014,Angelo2011,Smith2016,Loveridge2017,Loveridge2018,Giacomini2019,Giacomini2019a,Vanrietvelde2020,Hoehn2020,delaHamette2020,Ballesteros2021,Hamette2021,Hoehn2021,Krumm2021,Mikusch2021,Hoehn_2022,Ali_Ahmad_2022,Glowacki2023,Vanrietvelde2023,Hoehn2023,Hoehn2023a,de_la_Hamette_2023,Kabel_2024,Hausmann2025,Ballesteros2025,Hamette2025,Carette2025,DeVuyst2025a,DeVuyst2025b,carrozza2025,Araujoregado2025,Castro-Ruiz2025,Castroruiz2025comparisons,Garmier2025,Fiore2026,delahamette2026,rothlin2026,Luppi2026}, though without reaching a consensus~\cite{Castroruiz2025comparisons}. This lack of consensus stems from profound conceptual hurdles, one of which is the Paradox of the Third Particle~\cite{Angelo2011}. While the paradox arises across various frameworks, we present it, for concreteness, within what is termed the perspectival approach~\cite{Giacomini2019,delaHamette2020}.

Suppose that, from the perspective of $A$, a random bit $M$ is encoded in the spatial degree of freedom of the particle $B$. Specifically, when the value of the bit is $m \in \{0,1\}$, $B$ is set to a superposition of position eigenstates $\ket{\phi_m}_B \coloneq \frac{1}{\sqrt{2}}(\ket{1}_B + (-1)^m\ket{2}_B)$, as illustrated in the upper part of \cref{fig:quantum_two_particles}. This is described by the joint quantum state
\begin{equation}\label{eq:state_A}
    \rho_{BM}^{(A)} = \frac{1}{2} \sum_{m \in \{0,1\}} \ketbra{\phi_m}_B \otimes \ketbra{m}_M.
\end{equation}
Because the states $\ket{\phi_0}_B$ and $\ket{\phi_1}_B$ are orthogonal, the bit $M$ can be perfectly retrieved via a measurement on $B$.

What does this setup look like when viewed instead from the perspective of $B$? Geometrically, if $B$ is at position $x$ relative to $A$, then $A$ is at position $-x$ relative to $B$. In the perspectival approach, this intuition is applied within each branch of the position basis. Accordingly, the  quantum reference frame transformation $T^{A \to B}$ that yields $B$'s perspective from $A$'s is implemented by $T^{A \to B}: \ket{x}_B \mapsto \ket{-x}_A$, see \cref{fig:quantum_two_particles}. Because~$T^{A \to B}$ is an isometry, the joint state of $A$ and $M$ from $B$'s perspective, $\smash{\rho^{(B)}_{A M}}$, also possesses the form specified by~\cref{eq:state_A}. Consequently, the conclusion from above applies analogously: from the perspective of $B$, the bit $M$ can be retrieved via a measurement on~$A$.

\begin{figure}[tbp]
    \centering
    \begin{subfigure}{\textwidth}
        \centering     
        \includegraphics{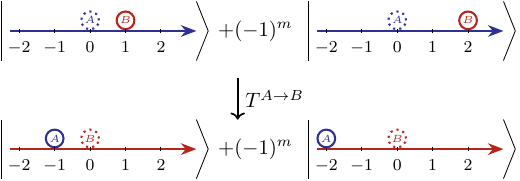}
        \caption{Perspective from particle $A$ on particle $B$ (top) and from particle $B$ on particle $A$ (bottom).}
        \label{fig:quantum_two_particles}
    \end{subfigure}

    \vspace{0.5cm}

    \begin{subfigure}{\textwidth}
        \centering     
        \includegraphics{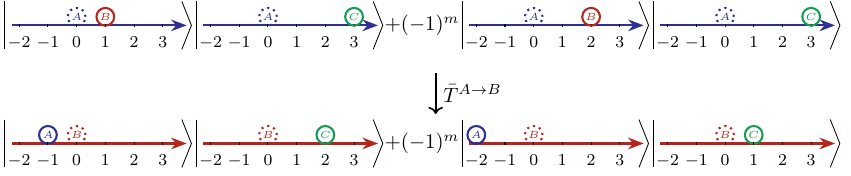}
        \caption{Perspective from particle $A$ (top) and particle $B$ (bottom) with the third particle $C$ included.}
        \label{fig:quantum_three_particles}
    \end{subfigure}
    \caption{\textbf{Quantum Paradox of the Third Particle.} From the perspective of one particle (dotted), the state of other particles (solid) depends on the value $m$ of a bit. From $A$'s perspective, the state of particle $B$ is the same independently of whether the third particle~$C$ is ignored, as in~(\subref{fig:quantum_two_particles}), or included in the description, as in~(\subref{fig:quantum_three_particles}). The paradox arises because, after transforming to $B$'s perspective (vertical arrow), the state of $A$ is different depending on whether $C$ is considered. This difference is operationally relevant: in (\subref{fig:quantum_two_particles}), the value $m$ is retrievable by acting on $A$ only, whereas this is not the case in~(\subref{fig:quantum_three_particles}). }\label{fig:quantum_paradox}
\end{figure} 

Returning to $A$'s perspective, we may add a third particle $C$ to the description, resulting in a joint state of the form
\begin{equation}\label{eq:state_A_ext}
    \rho_{BCM}^{(A)} = \frac{1}{2} \sum_{m \in \{0,1\}} \ketbra{\phi_m}_B \otimes \ketbra{3}_C \otimes \ketbra{m}_M, 
\end{equation}
as illustrated in the upper part of \cref{fig:quantum_three_particles}. Note that the marginal $\rho_{B M}^{({A})}$ is identical to the state defined by~\cref{eq:state_A}, so that $A$'s view of $B$ and $M$ remains unaltered. In particular, from $A$'s perspective, $M$ can still be retrieved by measuring $B$. 

Previously, we concluded that, from $B$'s perspective, $M$ is retrievable via a measurement on $A$. To perform a consistency check, we now investigate whether the same conclusion still holds within the extended description that includes the third particle. With $C$ explicitly in play, the reference frame transformation $\smash{\bar{T}^{A \to B}}$ to $B$'s perspective is given by 
\begin{equation}\label{eq:QRF_trafo}
    \bar{T}^{A \to B}: \quad \ket{x}_B \otimes \ket{x'}_C \ \mapsto \ \ket{-x}_A \otimes \ket{x'-x}_C.
\end{equation}
Analogous to $T^{A \to B}$, this transformation is justified by its geometric interpretation in the position eigenbasis (cf.\ \cref{fig:classical_RF_trafo} below). Applying it to the state in~\eqref{eq:state_A_ext} yields
\begin{equation}
    \rho_{ACM}^{(B)} = \frac{1}{2} \sum_{m \in \{0,1\}} \ketbra{\psi_m}_{AC} \otimes \ketbra{m}_M,
\end{equation}
where $\ket{\psi_m}_{AC} \coloneq \frac{1}{\sqrt{2}}\left(\ket{-1}_A \otimes \ket{2}_C + (-1)^m \ket{-2}_A \otimes \ket{1}_C \right)$ is an entangled state, as illustrated in the lower part of \cref{fig:quantum_three_particles}. To determine whether $M$ can be retrieved from $A$ alone, we trace out $C$, obtaining
\begin{equation}
    \tilde{\rho}_{AM}^{(B)} = \frac{1}{4}\Big(\ketbra{-1}_{A} + \ketbra{-2}_A\Big) \otimes \sum_{m \in \{0,1\}}  \ketbra{m}_M.
\end{equation}
Crucially, $A$ is independent of $M$. Thus, from $B$'s perspective, no information about $M$ can be obtained by measuring $A$, meaning that our consistency check failed!

This specific instance exemplifies the \emph{Paradox of the Third Particle}, which we characterize more generally by the non-commutative diagram in \cref{fig:paradox_diagram}. The paradox remains actively debated~\cite{Observers2026}, despite several proposed resolutions~\cite{Angelo2011,Krumm2021,Krumm2021unpub,Brukner2024,Castro-Ruiz2025}.

To better understand a paradox formulated within quantum theory, it is instructive to ask whether the paradox persists in a classical formulation~\cite{Spekkens_2007,Spekkens2015,Catani_2023,Hausmann_2023,Jones2026}. This question is particularly relevant here, since quantum theory introduces significant conceptual difficulties when multiple perspectives are involved~\cite{Nurgalieva_2020,Hausmann2025fire}. We show that the answer is yes: the Paradox of the Third Particle survives the removal of all its quantum features.

\begin{figure}[tbp]
    \centering
    \includegraphics{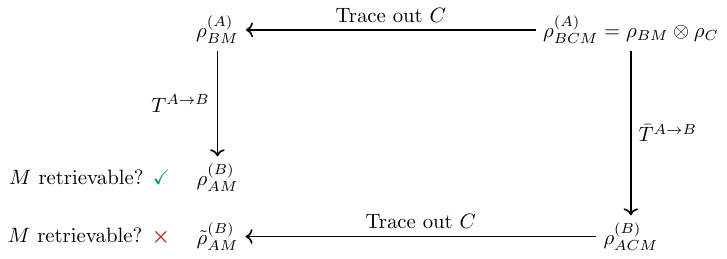}
    \caption{\textbf{Structure of the Paradox of the Third Particle.} From the perspective of particle $A$ (upper part), information $M$ can be retrieved from particle $B$ alone, i.e., without acting on the third particle $C$. An analysis of the same setup after transforming to the perspective of $B$ (lower part), however, yields contradictory answers to whether $M$ can be retrieved without access to~$C$. For concreteness, the diagram is phrased in terms of density operators. Its structure, however, does not depend on this particular representation of states and applies equally to non-quantum theories, such as classical mechanics.}\label{fig:paradox_diagram}
\end{figure}

\section{Reference frame transformations for classical particles in \texorpdfstring{$1D$}{1D}}
Before introducing the classical analogue of the Paradox of the Third Particle, we discuss classical reference frame transformations, taking into account the paradigm that reference frames are physical systems. This goes slightly beyond standard textbook treatments, where reference frames are regarded as abstract idealizations. To mirror the original quantum version of the paradox in our classical analogue, we formulate the underlying reference frame transformations via a set of postulates that retain their meaning outside the framework of quantum theory. While the concrete nature of these postulates renders them somewhat restrictive, we lift these limitations in \cref{sec:general_paradox}.

\begin{figure}
    \centering
    \includegraphics{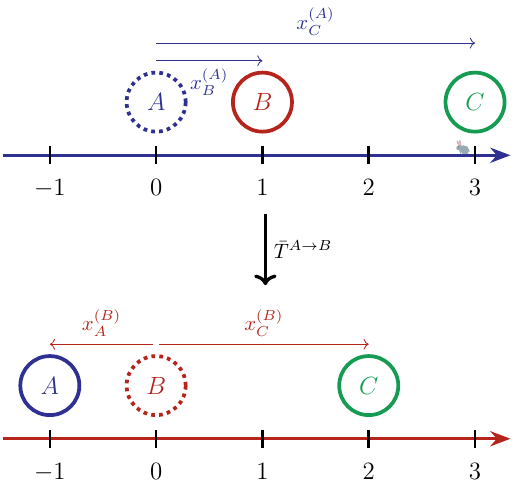}
    \caption{\textbf{Geometric intuition.} The upper part shows the perspective from particle $A$ (dotted) on particles $B$ and $C$ (solid). The lower part shows the perspective from particle $B$ (dotted) on particles $A$ and $C$ (solid). The reference frame transformation $\bar{T}^{A \to B}$ from $A$'s perspective to $B$'s perspective corresponds to a shift of the origin from the location of $A$ to the location of $B$.}
    \label{fig:classical_RF_trafo}
\end{figure}

To characterize the transformation between perspectives, we must first clarify what data is required to specify a perspective. In the quantum case, the perspective of $A$ on~$B$ and additional particles $C_1,\dots, C_n$ is given by a density operator on their joint state space, which contains sufficient information to determine the time evolution. Requiring that the classical analogue possess the same predictive power, a perspective must specify not merely positions, but also momenta. Formally, the classical perspective of $A$ consists of a phase space $\mathcal{P}^{(A)} = \mathbb{R}^{2(n+1)}$, which is endowed with the symplectic form $\Omega^{(A)} = dx^{(A)}_{B} \wedge dp^{(A)}_{B}+ \sum_{i = 1}^n dx^{(A)}_{C_i} \wedge dp^{(A)}_{C_i}$, and similarly for $B$.\footnote{See \cref{app:class_mech} for a review of phase spaces.} 

We are now ready to formulate the postulates governing the transformation from $A$'s perspective to $B$'s. The first postulate captures the geometric intuition underlying a quantum reference frame transformation like \cref{eq:QRF_trafo}, which is illustrated by \cref{fig:classical_RF_trafo}. 

\begin{postulatebox}
    \begin{postulate}[Geometry]\label{pos:geometry}
        If, from $A$'s perspective, $B$ is at position $x_B$ and $C_i$ at position $x_{C_i}$, then, from $B$'s perspective, $A$ is at position~$-x_B$ and $C_i$ at position $\smash{x_{C_i} - x_{B}}$.
    \end{postulate}
\end{postulatebox}

Another key feature of the quantum transformation defined by \cref{eq:QRF_trafo} is that it introduces no phase factors. Our second postulate serves as the classical analogue to this property.

\begin{postulatebox}
    \begin{postulate}[Position and momentum separation]\label{pos:pos_mom_sep}
        The momentum of any particle from $B$'s perspective is a function solely of the particles' momenta from $A$'s perspective. Furthermore, the phase space origin $(0,\dots,0)$ maps to itself.\footnotemark
    \end{postulate}
\end{postulatebox}
\footnotetext{\label{foot:implied}For the position coordinates, this preservation of the origin is already guaranteed by \cref{pos:geometry}.}

Finally, the quantum reference frame transformation \cref{eq:QRF_trafo} is unitary, which may be interpreted as the requirement that, although taking different perspectives, we are describing the same physical setup and possess identical information about it. 

\begin{postulatebox}
    \begin{postulate}[Information preservation]\label{pos:reversibility}
        The transformation from the perspective of $A$ to the perspective of $B$ is a symplectomorphism.
    \end{postulate}
\end{postulatebox}

Together, these three postulates completely characterize the classical reference frame transformation from $A$ to $B$.\footnote{The proofs of this and all subsequent statements are provided in \cref{app:proofs}.}

\begin{resultbox}
    \begin{restatable}[]{lemma}{classicalRF}\label{thm:classicalRF}
        \Cref{pos:geometry,pos:pos_mom_sep,pos:reversibility} uniquely fix the reference frame transformation to be
        \begin{align}
            \begin{split}
                \bar{T}^{A \to B}: \mathcal{P}^{(A)} &\to \mathcal{P}^{(B)} \\
                \left(\begin{psmallmatrix}
        x^{(A)}_B\\
        p^{(A)}_B
        \end{psmallmatrix}, \begin{psmallmatrix}
        x^{(A)}_{C_1}\\
        p^{(A)}_{C_1}
        \end{psmallmatrix}, \dots\right) &\mapsto \left(\begin{psmallmatrix}
        -x^{(A)}_B\\
        -p^{(A)}_B-\sum_i  p^{(A)}_{C_i}
        \end{psmallmatrix}, \begin{psmallmatrix}
        x^{(A)}_{C_1}-x^{(A)}_B\\
        p^{(A)}_{C_1}
        \end{psmallmatrix}, \dots \right).
            \end{split}
        \end{align}
    \end{restatable}
\end{resultbox}

The transformation in \cref{thm:classicalRF} is the classical analogue of quantum reference frame transformations like \cref{eq:QRF_trafo} introduced during our discussion of the Paradox of the Third Particle.\footnote{A reader bothered by how this transformation acts on the momenta is referred to the paragraph ``The basis problem'' in \cref{sec:classical_analogue}.} We are thus ready to formulate the classical version of this paradox. 

\section{The classical Paradox of the Third Particle}

We start with the perspective of a particle $A$ on a particle $B$. By analogy with the quantum case, a random bit $M$ is encoded into the momentum of $B$: if the bit takes the value $m$, the momentum of $B$ is likewise $m$ (in appropriate units). This is described by the joint probability distribution\footnote{For the pedantic, this defines a probability measure, with $\delta$ the Dirac measure. We use Dirac measures for simplicity only; in \cref{app:robust} we show that the paradox is robust under small deviations.} over the phase space of particle $B$ and the bit $M$, 
\begin{equation}
    P^{(A)}_{BM}\left(\begin{psmallmatrix}
    x_B\\
    p_B
    \end{psmallmatrix},m\right) = \frac{1}{2} \delta(x_B)\delta(m-p_B),
\end{equation}
which is illustrated in the left part of \cref{fig:classical_two_particles}. Evidently, the bit $M$ can be retrieved by measuring $B$. 

We now switch to $B$’s perspective using \cref{thm:classicalRF}, applying it without including any particles other than $A$ and $B$. We denote this special case by $T^{A \to B}$. The resulting distribution is shown in the right part of \cref{fig:classical_two_particles}. As is evident from the illustration, the bit $M$ can be retrieved by measuring $A$.

Next, we return to $A$'s perspective and include another particle $C$ in the description. We assume this particle is known to be at position $3$, while its momentum is distributed according to $\frac{1}{2K}\chi_{[-K,K]}$, the uniform measure over the interval $[-K,K]$ for some $K \in \mathbb{R}_{> 0}$, leading to the joint distribution
\begin{equation}\label{eq:classical_perspective_A}
    P^{(A)}_{BCM}\left(\begin{psmallmatrix}
    x_B\\
    p_B
    \end{psmallmatrix},\begin{psmallmatrix}
    x_C\\
    p_C
    \end{psmallmatrix},m\right) = \frac{1}{4K} \delta(x_B)\delta(m-p_B)\delta(x_C-3) \chi_{[-K,K]}(p_C),
\end{equation}
which is illustrated in \cref{fig:perspective_A}. Since the reduced distribution of particle~$B$ and the bit~$M$ remains unchanged, $M$ can still be retrieved by measuring $B$.

\begin{figure}[tb]
    \centering
    \begin{subfigure}{\textwidth}
        \centering
        \includegraphics[scale = 0.95]{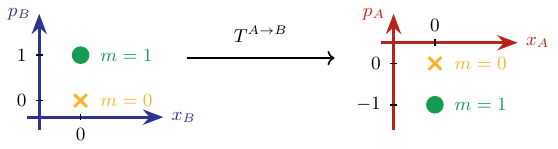}
        \caption{Perspective from particle $A$ on particle $B$ (left) and from particle $B$ on particle $A$ (right).}
        \label{fig:classical_two_particles}
    \end{subfigure}
    \begin{subfigure}{0.45\textwidth}
        \centering
        \includegraphics[scale = 0.95]{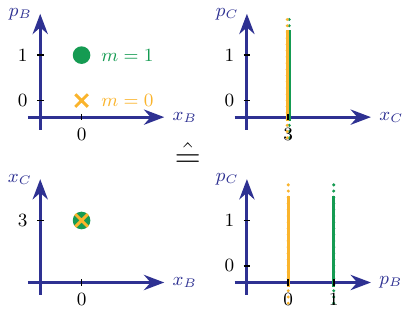}
        \caption{Perspective from particle $A$ with the third particle $C$ included.}
        \label{fig:perspective_A}
    \end{subfigure}
    \hspace{1cm}
   \begin{subfigure}{0.45\textwidth}
        \centering     
        \includegraphics[scale = 0.95]{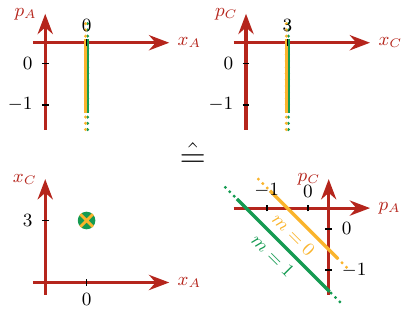}
        \caption{Perspective from particle $B$ with the third particle $C$ included.} 
        \label{fig:perspective_B}
    \end{subfigure}
    \caption{\textbf{The classical Paradox of the Third Particle.} From $A$'s perspective, $B$'s state depends on the value $m$ of the bit $M$, regardless of whether the third particle~$C$ is ignored, as in~(\subref{fig:classical_two_particles}), or included, as in~(\subref{fig:perspective_A}). However, upon transforming to $B$'s perspective, a discrepancy arises: the state of $A$ depends on $m$ if $C$ is ignored when transforming, as on the right-hand side of~(\subref{fig:classical_two_particles}), whereas it is independent of $m$ if $C$ is accounted for, as in~(\subref{fig:perspective_B}).}\label{fig:classical_paradox}
\end{figure} 

We once again switch to $B$’s perspective by applying \cref{thm:classicalRF}, this time explicitly including particle $C$ in the description. The resulting distribution, illustrated in \cref{fig:perspective_B}, is given by
\begin{equation}\label{eq:classical_perspective_B}
    P^{(B)}_{ACM}\left(\begin{psmallmatrix}
    x_{A}\\
    p_{A}
    \end{psmallmatrix},\begin{psmallmatrix}
    x_{C}\\
    p_{C}
    \end{psmallmatrix},m\right) = \frac{1}{4K} \delta(-x_A)\delta(m+p_A+p_C)\delta(x_C-x_A-3)\chi_{[-K,K]}(p_C).
\end{equation}
Notice that marginalizing over particle $C$ results in the reduced distribution
\begin{align}
    \nonumber
    P^{(B)}_{AM}\left(\begin{psmallmatrix}
    x_{A}\\
    p_{A}
    \end{psmallmatrix},m\right) &= \int dp_C dx_C \, \frac{1}{4K} \delta(-x_A)\delta(m+p_A+p_C)\delta(x_C-x_A-3) \chi_{[-K,K]}(p_C) \\ \label{eq:notation}
    &= \frac{1}{4K} \delta(-x_A)\chi_{[-K,K]}(-m-p_A).
\end{align}
We can now perform our consistency check and ask whether, from $B$'s perspective, the bit~$M$ is retrievable by acting solely on $A$. Looking at \cref{eq:notation} for large $K$, the distribution is close to $\frac{1}{4K} \delta(-x_A)\chi_{[-K,K]}(-p_A)$, and hence effectively independent of $M$. Therefore, retrieval is impossible. The consistency check thus fails, just as in the quantum case.

In summary, we have reproduced the non-commutative structure illustrated in \cref{fig:paradox_diagram} within classical mechanics. This establishes the claim in the title of this paper.

\section{Can the Paradox of the Third Particle be avoided?}\label{sec:general_paradox}
The classical version of the Paradox of the Third Particle relies on the specific reference frame transformations determined by \cref{pos:geometry,pos:pos_mom_sep,pos:reversibility}. Given that these postulates are quite restrictive, one may ask whether the paradox can be avoided by relaxing them. In what follows, we establish a no-go theorem demonstrating that this is not the case; rather, the paradox is merely a symptom of a deeper problem.

To keep the theorem's assumptions minimal, we define a \emph{perspective} by three ingredients only: the system $A$ that serves as the physical reference frame, the set $\mathbb{S}$ of systems described from the perspective, as reflected in notation like $\rho^{(A)}_{BC}$ with $\mathbb{S} = \{B,C\}$, and the joint state space $\mathcal{D}$ of these described systems.

\begin{postulatebox}
    \begin{definition}\label{def:perspective}
        A \emph{perspective} is a tuple $\smash{\mathbb{S}_{\mathcal{D}}^{(A)} \coloneq (A, \mathbb{S}, \mathcal{D})}$, where $A$ is a system label, $\mathbb{S}$ a set of system labels such that $A \notin \mathbb{S}$, and $\mathcal{D}$ is a state space. 
    \end{definition}
\end{postulatebox}
 
The choice of state space depends on the underlying physical theory. In quantum mechanics, the state space is given by the density matrices over a Hilbert space,\footnote{In the following, we assume that all Hilbert spaces are separable.} defining a \emph{quantum perspective}. Conversely, in classical mechanics, the state space is given by probability distributions over a phase space, defining a \emph{classical perspective}.

The second ingredient needed to formulate our no-go theorem is the notion of a reference frame transformation, which is defined as a map between the state spaces of different perspectives.

\begin{postulatebox}
    \begin{definition}\label{def:trafo}
        A \emph{reference frame transformation} from a perspective $\smash{\mathbb{S}^{(A)}_{\mathcal{D}}}$ to a perspective~$\smash{\bar{\mathbb{S}}^{(B)}_{\bar{\mathcal{D}}}}$ is a map $T: \mathcal{D} \to \bar{\mathcal{D}}$. 
    \end{definition}
\end{postulatebox}

The literature mostly considers transformations between perspectives~$\mathbb{S}^{(A)}_{\mathcal{D}}$ and $\bar{\mathbb{S}}^{(B)}_{\bar{\mathcal{D}}}$ that are \emph{maximally overlapping}, meaning the symmetric difference satisfies $\mathbb{S} \triangle \bar{\mathbb{S}} = \{A, B\}$. This ensures that each frame views the other as a physical system, while the two perspectives otherwise describe the same set of systems.

\begin{figure}
    \centering
    \includegraphics{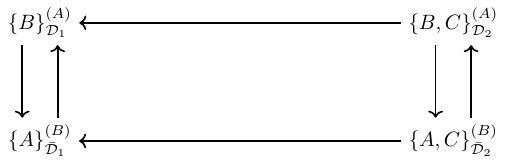}
    \caption{\textbf{Atlas of perspectives.} An atlas of perspectives is a directed graph whose nodes represent individual perspectives and whose edges define transformations from one to the other. This figure illustrates a subgraph where the Paradox of the Third Particle arises. Vertical edges denote reference frame transformations between $A$'s and $B$'s perspectives. Horizontal edges correspond to removing a system from the description without changing the reference.}
    \label{fig:no_go_framework}
\end{figure}

An \emph{atlas of perspectives} is a directed graph whose nodes are perspectives and whose edges are reference frame transformations between their state spaces, as illustrated in \cref{fig:no_go_framework}. An obvious consistency requirement we impose is that any two paths from a perspective $\mathbb{S}^{(A)}_{\mathcal{D}}$ to a perspective $\bar{\mathbb{S}}^{(B)}_{\bar{\mathcal{D}}}$ yield identical composed maps. We call an atlas \emph{non-trivial} if it contains \cref{fig:no_go_framework} as a subgraph.

Within a perspective, the described systems enter only as abstract labels, not yet tied to any part of the state space. To give these labels meaning, we introduce a \emph{subsystem partitioning}, which specifies how the state space decomposes into the labelled subsystems. For a quantum perspective, this corresponds to endowing the underlying Hilbert space with a tensor product structure whose tensor factors carry the subsystem labels. Similarly, for a classical perspective, it means equipping the underlying phase space $\mathcal{P}$ with a Cartesian product structure, whose factors carry the corresponding labels (see \cref{app:class_mech}). For example, if there are two subsystems $B$ and $C$, we decompose $\mathcal{P}$ into the Cartesian product of two symplectic manifolds $\mathcal{P} = \mathcal{P}_B \times \mathcal{P}_C$. 

We can now state our no-go theorem, which holds in both a classical and a quantum version.
\begin{resultbox}
    \begin{restatable}[]{theorem}{nogo}\label{thm:nogo}
        Consider a non-trivial atlas of classical (quantum) perspectives together with a subsystem partitioning of every perspective. Then, at least one of the following properties must be violated:
        \begin{itemize}[itemsep=0pt, topsep=3pt]
            \item {\textbf{Information preservation:}} Reference frame transformations between any two maximally overlapping perspectives are induced by symplectomorphisms (unitaries) between the underlying phase spaces (Hilbert spaces).
            \item {\textbf{Relationality:}} Reference frame transformations depend non-trivially on the target frame's state, i.e., do not factorize with respect to the subsystem partitioning.
            \item {\textbf{Decomposability:}} The removal of a system from a perspective corresponds to a trace with respect to the subsystem partitioning. 
        \end{itemize}
    \end{restatable}
\end{resultbox}

Importantly, the no-go theorem applies to an arbitrary choice of subsystem partitioning within each perspective. The non-trivial requirement is that both relationality and decomposability refer to the same choice. 

The no-go theorem tells us that the Paradox of the Third Particle has no easy resolution. To see this, note that any resolution of the paradox turns the diagram in \cref{fig:paradox_diagram} into a commutative one, making it a non-trivial atlas of perspectives. Therefore, the resolution must violate one of the three properties stated in \cref{thm:nogo}. 

\paragraph{Information preservation.} This condition is satisfied by most approaches to quantum reference frames, including the perspectival \cite{Giacomini2019,delaHamette2020}, as well as what are known as the perspective-neutral \cite{Vanrietvelde2020,Hoehn2021,Hamette2021} and the extra-particle approaches \cite{Castro-Ruiz2025,Garmier2025}. It is also a necessary condition for extended symmetry principles \cite{Hardy2018,Zych2018,Giacomini2019,Giacomini2020equivalence,Hardy2020,Giacomini2022,de_la_Hamette_2023,Kabel_2024}, which postulate that physical laws must be invariant under quantum reference frame transformations. However, frameworks exist where reference frame changes are not information preserving \cite{Palmer2014,Carette2025}.

While information preservation appears self-evident when treating reference frames as abstract coordinate systems, this intuition may fail when they are modelled as physical systems. The following example illustrates this.
\begin{resultbox}
    \begin{example}\label{ex:different_subsystems}
        Consider an external reference frame, such as the Earth's surface. From this external perspective, three particles $A$, $B$, and $C$ possess positions and canonically conjugate momenta $x_A,p_A,x_B,p_B,x_C,p_C$. Suppose the perspectives of~$A$ and $B$ are defined by
        \begin{align}\label{eq:extra_particle}
            \begin{split}
                x_B^{(A)} &= x_B - x_A ,\quad p_B^{(A)} = p_B \\
                x_C^{(A)} &= x_C - x_A, \quad p_C^{(A)} = p_C
            \end{split}
            \begin{split}
                x_A^{(B)} &= x_A - x_B, \quad p_A^{(B)} = p_A \\
                x_C^{(B)} &= x_C - x_B, \quad p_C^{(B)} = p_C.
            \end{split}
        \end{align}
        We can interpret momenta as generators of spatial translations. For instance,~$p_B^{(A)}$ shifts the position of particle $B$ from the perspective of $A$. Because $p_B^{(A)}$ Poisson-commutes with $x_A$, moving particle $B$ leaves $A$ fixed relative to the Earth's surface. Conversely, when $B$ moves $A$, it is $B$ that remains stationary. The key feature of this example is that $A$'s perspective includes $p_B$, i.e., $B$ can move relative to the Earth's surface, whereas this is not the case when looking at the world from $B$'s perspective. Consequently, the perspectives of $A$ and $B$ describe different subsystems of the external perspective, implying that the transformation between them cannot be information-preserving. 
    \end{example}
\end{resultbox}

\paragraph{Relationality.} The purpose of a reference frame transformation is to translate a description relative to $A$ into a description relative to $B$. This relativity implies that the resulting description of any system~$C$ relative to $B$ must depend on the state of $B$ relative to $A$. This dependence is precisely what relationality formalizes. We are unaware of any approach in the literature that violates this property.

\paragraph{Decomposability.}
Consider the task of predicting the orbit of the planets around the Sun. One could either model just the solar system or also include the distant black hole Sagittarius A${}^{*}$ into the description. Because the interaction with the black hole is irrelevant to the orbit of the planets around the Sun, both approaches yield the same result. If this were not the case, we could only do physics by describing the entire universe, which is obviously impossible. This illustrates a fundamental requirement for any physical theory: it must be applicable to subsystems of the universe. In particular, there should be a consistent way to remove an irrelevant system from a description.

The resolutions we present in the following either retain this system removal rule while giving up the requirement that decomposability and relationality refer to the same subsystem partitioning, or modify the removal rule itself.\footnote{Modified removal operations may be benchmarked against two properties of the trace-based rule: upon removing a system $C$ from $BC$, it maps every state on $BC$ to a state on $B$, and every possible state on $B$ arises as such an image.}
A first example is the resolution proposed in~\cite{Angelo2011}, where the Paradox of the Third Particle was introduced.

\begin{resolution}{Original formulation \cite{Angelo2011}}{original}
    Consider an external perspective as in \cref{ex:different_subsystems} and assume the centre of mass of the total system is conserved. We simplify the argument of \cite{Angelo2011} by assuming that $A$ and $B$ are of equal mass and much heavier than~$C$. Motivated by this physical picture, the subsystem partitionings for the perspectives of $A$ and $B$ are defined as: 
    \begin{align}
    \begin{split}
    x_B^{(A)} &= x_B - x_A, \quad p_B^{(A)} = \frac{1}{2} (p_B-p_A) \\
    x_C^{(A)} &= x_C - x_A, \quad p_C^{(A)} = p_C
    \end{split}
    \begin{split}
    x_A^{(B)} &= x_A - x_B, \quad p_A^{(B)} = \frac{1}{2} (p_A - p_B) \\
    x_C^{(B)} &= x_C - x_B, \quad p_C^{(B)} = p_C
    \end{split}
    \end{align}
    The lack of commutativity between $p_B^{(A)}$ and $x_C^{(A)}$ is identified as the root of the Paradox of the Third Particle. Physically, the non-commutativity arises because the momentum~$p_B^{(A)}$ effectively shifts the position of particle $A$, thereby altering $\smash{x_C^{(A)}} = x_C - x_A$. Consequently, particles $B$ and $C$ are not independent subsystems from $A$'s perspective. In the terminology of \cref{thm:nogo}, this resolution gives up decomposability with respect to the above subsystem partitionings.
\end{resolution}

We now present a second example for a resolution, which applies within the perspectival approach. It follows the ideas proposed in~\cite{Brukner2024,Krumm2021unpub}.

\begin{resolution}{Perspectival approach}{}
   The argument begins by noting that, from the perspective of $A$, for any observable $O_B$ on $B$, $\smash{\trace(\rho^{(A)}_{B}O_{B} ) = \trace(\rho^{(A)}_{BC}O_{B}\otimes \mathbbm{1}_C)}$.
    Applying the reference frame transformation to the state and the observables on both sides of the equation yields 
    \begin{equation}
        \tr(\rho^{(B)}_{A} T^{A \to B} O_{B} (T^{A \to B})^{\dagger}) = \tr(\rho^{(B)}_{AC} \bar{T}^{A \to B} (O_{B}\otimes \mathbbm{1}_C) (\bar{T}^{A \to B})^{\dagger}).
    \end{equation}
    Thus, from $B$'s perspective, the observable $T^{A \to B} O_{B} (T^{A \to B})^{\dagger}$ on $A$ has the same expectation value as $\bar{T}^{A \to B} (O_{B}\otimes \mathbbm{1}_C) (\bar{T}^{A \to B})^{\dagger}$.
    This motivates the definition of the subsystem~$A$ of the total system $AC$ as the one induced by the set of all observables of the form $\bar{T}^{A \to B} (O_{B}\otimes \mathbbm{1}_C) (\bar{T}^{A \to B})^{\dagger}$. While decomposability holds with respect to such a subsystem partitioning, relationality fails, because $\bar{T}^{A \to B}$ factorizes relative to it. This problem may be circumvented by giving up the requirement that relationality and decomposability refer to the same subsystem partitioning.
\end{resolution}

The perspective-neutral approach to quantum reference frames~\cite{Vanrietvelde2020,Hamette2021} offers an alternative resolution that maintains information preservation by modifying decomposability.

\begin{resolution}{Perspective-neutral approach \cite{Krumm2021}}{}
    This resolution proposes replacing the trace with an alternative operation called the \emph{relational trace}. Our no-go theorem implies that the relational trace cannot share all properties of the standard trace. Indeed, it is not trace-preserving. This challenges the interpretation of probabilities,\textsuperscript{\ref{foot:probabilities}} although yielding the correct expectation values for relational observables. Furthermore, as tracing a system $C_1$ and subsequently another system $C_2$ does not result in the same state as when both~$C_1$ and~$C_2$ are traced out simultaneously, the relational trace does not yield a consistent atlas of perspectives. 
\end{resolution}

\refstepcounter{footnote}\label{foot:probabilities}\footnotetext{Consider our $1$-dimensional setup with particles $A$, $B$, and $C$, which we assume to be at different positions. Provided that the density operator including $C$ is normalized, \cite{Krumm2021} gives results compatible with $P(x_A \leq x_B \land x_A \leq x_C) + P(x_A \leq x_B \land x_A \geq x_C) + P(x_A \geq x_B \land x_A \leq x_C) + P(x_A \geq x_B \land x_A \geq x_C) = 1$. However, if we apply the relational trace to remove $C$ from our description then, due to the failure of trace preservation, a calculation of $P(x_A \leq x_B)$ and $P(x_A \geq x_B)$ yields values that do not generally sum up to $1$.}

Lastly, we consider the extra-particle approach~\cite{Castro-Ruiz2025,Garmier2025}.

\begin{resolution}{Extra-particle approach \cite{Castro-Ruiz2025}}{}
    In this approach, the perspective of a reference frame on another system is defined by a procedure that starts from an external perspective, and, crucially, is independent of what other systems are part of the description. However, to define information-preserving reference frame transformations, one needs to enlarge each perspective with another system---the ``extra particle''. The state of this extra particle depends on what systems are part of the description. Therefore, any removal operation that leads to a consistent atlas of perspectives must act on the extra particle. In particular, the removal operation cannot correspond to a trace, which would be required by our decomposability assumption.
\end{resolution}

\section{What quantum reference frame phenomena have a classical analogue?}\label{sec:classical_analogue}
In this section, we show that several phenomena commonly attributed to quantum reference frames, besides the Paradox of the Third Particle, also arise within classical mechanics.

\paragraph{Subsystem relativity.}
One may ask whether the subsystem partitioning changes if one transforms from one reference frame to another. This question was studied in~\cite{Ali_Ahmad_2022} for the perspective-neutral approach. For concreteness, take the perspective of a particle $A$ that describes two particles, $B$ and~$C$, and let $\mathcal{A}_C^{(A)}$ be the algebra of operators on the system~$C$, i.e., operators of the form $\smash{\mathbbm{1}^{(A)}_B \otimes O^{(A)}_C}$. The subsystem partitioning is said to be \emph{relative} if 
\begin{equation} \label{eq:subsystemrelativity}
    \bar{T}^{A \to B}\mathcal{A}_C^{(A)}(\bar{T}^{A \to B})^{\dagger}\neq \mathcal{A}_C^{(B)}.
\end{equation}
This holds, in particular, for quantum reference frame transformations in the perspective-neutral approach~\cite{Ali_Ahmad_2022}.  

In the classical case, the operators on $C$ are functions $f \in \mathcal{C}^{\infty}(\mathcal{P}_B^{(A)} \times \mathcal{P}_C^{(A)})$ of the form\footnote{The classical form is close to the quantum one if one notices that such $f$ can also be written as $f = \mathbbm{1} \times g$, with $\mathbbm{1}$ the function that is $1$ everywhere on $\mathcal{P}_B^{(A)}$.} $f(z_B,z_C) = g(z_C)$, where $\smash{z_B \in \mathcal{P}_B^{(A)}}$ and $\smash{z_C \in \mathcal{P}_C^{(A)}}$.
Furthermore, \cref{eq:subsystemrelativity} translates to
\begin{equation}
    \{f\circ (\bar{T}^{A \to B})^{-1}: f \in \mathcal{A}_C^{(A)}\} \neq \mathcal{A}_C^{(B)}.
\end{equation}

The following theorem shows that subsystem relativity is a generic feature.\footnote{Note that relativity of subsystems poses a challenge for information theory and applications like cryptography across perspectives. For example, whether a key is secure against an adversary with access to a system $E$ could depend on the perspective, as what constitutes the subsystem~$E$ is perspective dependent.}
\begin{resultbox}
    \begin{restatable}[]{theorem}{subsysrelativity}\label{thm:subsysrelativity}
        In both quantum theory and classical mechanics, if information preservation holds, then for any subsystem partitioning:
        \begin{center}
            relationality of reference frame transformations

            $\iff$ 

            relativity of subsystem partitioning
        \end{center} 
    \end{restatable}
\end{resultbox}
Together with \cref{thm:nogo}, this result implies that, assuming information preservation, if the subsystem partitioning is relative then decomposability does not hold. 

\paragraph{Relativity of correlations.} Whether two subsystems are correlated generally depends on the quantum reference frame from which one describes them~\cite{Giacomini2019}. This feature is an immediate consequence of the relativity of the subsystem partitioning. It thus follows from the above that this feature also occurs classically. The relativity of correlations is also apparent in the setup we considered for our classical version of the Paradox of the Third Particle. There, the system $C$ was uncorrelated from $A$'s perspective (eq.~\cref{eq:classical_perspective_A}), but correlated with $A$ from $B$'s perspective (eq.~\cref{eq:classical_perspective_B}). 

\paragraph{The basis problem.}
The reference frame transformation of \cref{thm:classicalRF} is only relative in terms of position, not momentum, contrary to the intuition that momentum should be relative as well. In other words, the particles serve as reference frames only for translations, rather than both translations and boosts. However, \cref{pos:geometry,pos:pos_mom_sep,pos:reversibility} should also be satisfied by reference frame transformations that include boosts. Since the transformation provided by \cref{thm:classicalRF} is unique given these postulates, one cannot impose relativity in momentum without giving up relativity in position. This dilemma is the ``basis problem''. It occurs in the perspectival approach, both classically and quantumly, as recognized already in~\cite{Giacomini2019}.

\paragraph{A single particle is not an ideal reference frame for translations and boosts.} In the literature, a particle is said to \emph{serve as an ideal\footnote{Sometimes the term \emph{perfect} is used instead of \emph{ideal}.} reference frame} for a group $G$ if it can perfectly encode an element of $G$ \cite[Equation (4.24)]{Bartlett2007}. The motivation for this definition is that, according to quantum theory, an ideal frame allows one to construct a unitary that applies different unitaries controlled on the encoded group element. In particular, this is possible if the system that serves as a reference has the Hilbert space $L^2(G)$.

Somewhat surprisingly, in classical mechanics, this connection between encoding and symplectic controlled operations does not hold. To see this, we first note that, in contrast to a quantum particle, which is subject to uncertainty relations, a classical particle serves as an ideal reference frame for the group formed by translations and boosts, according to the above definition. Assume now, by contradiction, that it is possible to construct a symplectic transformation that applies different symplectic transformations controlled on the encoded group element. We could then construct a symplectic reference frame transformation from $A$'s to $B$'s perspective that shifts the position and momentum of the particles $C_i$ depending on the position and momentum of the particle $B$ from $A$'s perspective. We would thus have achieved a reference frame transformation into $B$'s perspective for both translations and boosts. This is impossible, as we noted in the above discussion of the basis problem. 

This argument shows that, classically, the above definition of ideality for reference frames is unsatisfactory. Instead, one should define \emph{ideality} directly via the ability to construct symplectic controlled operations. Then, according to the above argument,  a classical particle, like a quantum particle, is not an ideal reference frame for both translations and boosts. One may then ask which phase space a system must have in order to serve as an ideal reference frame for a Lie group $G$ in this sense. One can show that the cotangent bundle\footnote{See, for example, \cite[p.~63]{Libermann1987} for details on how to construct a symplectic structure on the cotangent bundle of a manifold.} $\mathcal{P} = G \times \mathfrak{g}^{*}$, where $\mathfrak{g}^{*}$ is the dual of the Lie algebra $\mathfrak{g}$ of the group $G$, provides such a phase space. This is the classical analogue of $L^2(G)$. 

\section{Conclusions}
Our result---that the Paradox of the Third Particle occurs classically---demonstrates that it is not an effect specific to quantum theory, but rather a general symptom of treating reference frames as physical systems. This can be understood in information-theoretic terms: a reference frame transformation satisfying all conditions of \cref{thm:nogo} would act as an ``information diode'', allowing information to flow from the target reference frame $B$ to system $C$ while leaving $B$ unaffected.\footnote{\label{ftn:classicaldisturbance}How is this compatible with the claim that classically a system can be measured without disturbing it? The answer is that this is only possible for particular states of the measurement apparatus. For other states, there is disturbance.} However, in both classical and quantum theory, such a diode is impossible: any non-trivial information-preserving interaction inevitably causes bidirectional information flow.\footnote{\Cref{thm:nogo} can be read as a formalization of precisely this statement. Relationality is a necessary condition for information being extracted from $B$. Furthermore, decomposability, together with the atlas being consistent, formalizes the requirement that no information flows from $C$ to $B$.}

This observation is also relevant when one operates on a system relative to a physical reference frame---for instance, when measuring the position of a particle relative to another particle. To analyse the effect of such an operation on the frame itself, we may describe both the system and the frame from the perspective of an additional, external frame, as in \cref{ex:different_subsystems}. From this external perspective, the operation couples the system to the frame and therefore entails a back-action on it: in general, the state of the frame changes, both classically and quantumly (cf.~\cref{ftn:classicaldisturbance}). This reflects a fundamental difference between physical reference frames and abstract coordinate systems, which by construction have no dynamical degrees of freedom.

Regarding reference frames as physical systems poses an additional conceptual challenge: a change of perspective alters the set of described systems, effectively shifting the ``Heisenberg cut'', as illustrated in \cref{fig:Heisenberg_cuts}. Reasoning that combines descriptions relative to different Heisenberg cuts is known to be delicate in quantum theory~\cite{Nurgalieva_2020,Hausmann2025fire}. Since reference frame transformations relate such descriptions, they inherit this difficulty.

\begin{figure}
    \centering
    \begin{subfigure}{0.48\textwidth}
        \centering
        \includegraphics[scale = 0.95]{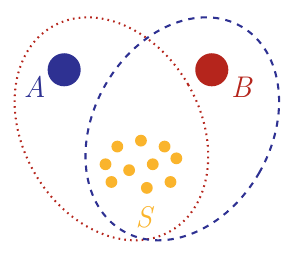}
        \caption{
            Heisenberg cuts of the perspective of $A$ and $B$ (blue dashed and red dotted, respectively) when one is modelled as a physical system in the other's perspective.}
        \label{fig:QRF_cut}
    \end{subfigure}
    \hfil
    \begin{subfigure}{0.48\textwidth}
        \centering
        \includegraphics[scale = 0.95]{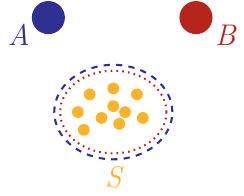}
        \caption{Heisenberg cuts of the perspective of $A$ and $B$ (blue dashed and red dotted, respectively) when neither is modelled as a physical system in the other's perspective.}
        \label{fig:normal_cut}
    \end{subfigure}
    \caption{\textbf{Heisenberg cuts.} Originally, the notion of a Heisenberg cut was introduced in quantum theory as the boundary between the world described by quantum theory and the rest. Here, we use this notion more generally: the boundary between the systems described from a perspective and the rest. Figure (\subref{fig:QRF_cut}) shows the case relevant to the Paradox of the Third Particle, and Figure (\subref{fig:normal_cut}) the standard situation.}
    \label{fig:Heisenberg_cuts}
\end{figure}

A further conceptual difficulty arises when frame transformations are interpreted as mappings between the epistemic states of different observers. Suppose Bob has his eyes open, while Alice has hers closed and assigns a probability of $\frac{1}{2}$ to Bob being at position $1$ or $100$. A standard quantum reference frame transformation implies that Bob must possess a matching uncertainty, locating Alice at $-1$ or $-100$ with equal probability. Yet, this directly contradicts Bob's actual local information: since his eyes are open, he knows Alice is at a definite position, say $-1$.

This problem occurs because Alice and Bob have different information about the setup. Reference frame transformations, however, fail to account for this difference. If they do not map between descriptions of different observers, how should reference frame transformations be interpreted? One option is to view them as transformations between descriptions held by a single observer with fixed knowledge, such as Alice, adapted to using different reference frames for her measurement apparatus. 

On a separate note, our construction of the classical Paradox of the Third Particle requires introducing uncertainty regarding a particle's position and momentum. This aligns well with the insights of Spekkens' toy theory \cite{Spekkens_2007,Spekkens2015,Catani2017,Hausmann2021}, which demonstrates that numerous quantum features find classical counterparts once one imposes a restriction, akin to the Heisenberg uncertainty relation, on what can be known about phase-space variables. 

Finally, we hope that our no-go theorem, \cref{thm:nogo}, is not read as ruling out transformations between physical reference frames. Instead, we see it as an invitation to explore their foundations. For example, the notion of the atlas of perspectives introduced here provides a playground for this exploration. It gives concrete criteria, such as consistency: different paths through any atlas, not only the one in \cref{fig:no_go_framework}, must lead to the same result.

\section*{Acknowledgements}
We thank Carlo Cepollaro, Thomas Galley, Fedele Lizzi, Maximilian Lock, Lorenzo Maccone, and T. Rick Perche for discussions. This work was funded by the Swiss National Science Foundation via project No.\ \mbox{20QU-1\_225171}. We are also grateful for the support from the NCCR SwissMAP, the ETH Zurich Quantum Center and the WithOut SpaceTime (WOST) project funded by the Grant ID\#~62312 from the JTF. The opinions expressed in this work are those of the authors and do not necessarily reflect the views of the JTF. This work is part of the European Union COST Actions Bridging high and low energies in search of quantum gravity (CA23130) and Relativistic Quantum Information (CA23115).
 
\printbibliography
\numberwithin{theorem}{section}
\numberwithin{definition}{section}
\numberwithin{lemma}{section}
\numberwithin{remark}{section}
\numberwithin{proposition}{section}
\numberwithin{corollary}{section}
\numberwithin{equation}{section}

\appendix
\newpage 

\section*{\LARGE{Appendix}}

\section{Classical mechanics}\label{app:class_mech}
In the Hamiltonian formulation of classical mechanics, the state space of the system is its phase space, which, mathematically, is a symplectic manifold. A map that preserves the symplectic structure is called a \emph{symplectomorphism}, i.e., a diffeomorphism $T: \mathcal{P} \to \mathcal{Q}$ such that $T^{*}\Omega_{\mathcal{Q}} = \Omega_{\mathcal{P}}$, where $\mathcal{P}, \mathcal{Q}$ are symplectic manifolds and $\Omega_{\mathcal{Q}}, \Omega_{\mathcal{P}}$ their respective symplectic forms.

For this paper, manifolds will always be connected and finite-dimensional. Furthermore, often they are $\mathbb{R}^{2N}$ and the symplectic form $\Omega$ is such that there are global coordinates $(\vec{q},\vec{p})$ in which $\Omega = \sum_{i = 1}^N dq_i \wedge dp_i$. This symplectic form induces Poisson brackets $\{\cdot, \cdot\}$ as follows\footnote{For the definition of Poisson brackets on more general symplectic manifolds see, for example, \cite[Section 14.7]{Libermann1987}.}
\begin{equation}
    \{f, g\} = \sum_i \frac{\partial f}{\partial q_i} \frac{\partial g}{\partial p_i} - \frac{\partial g}{\partial q_i} \frac{\partial f}{\partial p_i}.
\end{equation}
Operationally, the Poisson brackets determine how one can act on the system. Concretely, if the system evolves under the Hamiltonian $H$, then its state $(\vec{q},\vec{p})$ evolves as 
\begin{equation}
    \dot{q}_i = \{q_i,H\} \text{ and } \dot{p}_i = \{p_i,H\}.
\end{equation}

\subsection{The Cartesian product of phase spaces}
The Cartesian product of two manifolds $\mathcal{P}_B, \mathcal{P}_C$ is the set $\mathcal{P}_B \times \mathcal{P}_C$ equipped with the product topology. Its charts are given by the Cartesian product of the charts of $\mathcal{P}_B, \mathcal{P}_C$: if $\psi_B: U \to \mathbb{R}^m$ is a chart of $\mathcal{P}_B$ and $\psi_C: V \to \mathbb{R}^n$ is a chart of $\mathcal{P}_C$, then $\psi_B \times \psi_C: U \times V \to \mathbb{R}^{m+n}$ is a chart of $\mathcal{P}_B \times \mathcal{P}_C$. 

If $\mathcal{P}_B, \mathcal{P}_C$ are symplectic manifolds with symplectic forms $\Omega_B, \Omega_C$, respectively, then there is a natural way to define a symplectic form $\Omega_{BC}$ on $\mathcal{P}_B \times \mathcal{P}_C$ as
\begin{equation}
    \Omega_{BC} = \Pi_{B}^{*}\Omega_{B} + \Pi_{C}^{*}\Omega_{C},
\end{equation}
where $\Pi_{B}^{*}$ and $\Pi_{C}^{*}$ are the pullbacks of the canonical projections \cite[p.\ 21]{Silva2008}. These canonical projections can also be used to define the trace over $C$ of a probability measure~$P_{BC}$ on $\mathcal{P}_B \times \mathcal{P}_C$. For a measurable set $X$ on $\mathcal{P}_B$, we define
\begin{equation}
    \tr_C(P_{BC})(X) \coloneq P_{BC}(\Pi_{B}^{-1}(X)).
\end{equation}
Finally, we say that a symplectomorphism $T_{BC \to DE}: \mathcal{P}_{B} \times \mathcal{P}_C  \to \mathcal{P}_{D} \times \mathcal{P}_E$ \emph{factorizes} if there exist symplectomorphisms $T_{B \to D}: \mathcal{P}_{B} \to \mathcal{P}_{D}$ and $T_{C \to E}: \mathcal{P}_{C} \to \mathcal{P}_{E}$ such that $T_{BC \to DE} = T_{B \to D} \times T_{C \to E}$.

\section{Proofs}\label{app:proofs}

\subsection{Classical reference frame transformations}
\begin{resultbox}
    \classicalRF*
\end{resultbox}

\begin{proof}
    The symplectic form on $\mathcal{P}^{(A)}$ is given by 
    \begin{equation}
       \Omega^{(A)} = dx^{(A)}_B \wedge dp^{(A)}_B + \sum_i dx^{(A)}_{C_i} \wedge dp^{(A)}_{C_i}.
    \end{equation}
    Due to \cref{pos:pos_mom_sep}, we know that the push-forwards of ${\frac{\partial}{\partial x^{(A)}_B}}$ and ${\frac{\partial}{\partial x^{(A)}_{C_i}}}$ under $\bar{T}^{A \to B}$ are both linear combinations of $\smash{\frac{\partial}{\partial x^{(B)}_A}}$ and $\smash{\frac{\partial}{\partial x^{(B)}_{C_i}}}$. \Cref{pos:geometry} then implies that 
    \begin{equation}
        \bar{T}^{A \to B}_{*} \frac{\partial}{\partial x^{(A)}_B} = -\frac{\partial}{\partial x^{(B)}_A} - \sum_{i} \frac{\partial}{\partial x^{(B)}_{C_i}} \text{ and } \bar{T}^{A \to B}_{*}  \frac{\partial}{\partial x^{(A)}_{C_i}} =  \frac{\partial}{\partial x^{(B)}_{C_i}}.
    \end{equation}
   Furthermore, \Cref{pos:geometry} also implies that
    \begin{equation}
        \bar{T}^{A \to B}_{*} \frac{\partial}{\partial p^{(A)}_B} = c^{(B)}_A \frac{\partial}{\partial p^{(B)}_A} + \sum_i c^{(B)}_{C_i} \frac{\partial}{\partial p^{(B)}_{C_i}}
    \end{equation}
    and
    \begin{equation}
        \bar{T}^{A \to B}_{*} \frac{\partial}{\partial p^{(A)}_{C_j}} = c^{(C_j)}_A \frac{\partial}{\partial p^{(B)}_A} + \sum_i c^{(C_j)}_{C_i} \frac{\partial}{\partial p^{(B)}_{C_i}}.
    \end{equation}
    Using \cref{pos:reversibility}, we find that  
    \begin{align}
         \Omega^{(B)}(\bar{T}^{A \to B}_{*} \frac{\partial}{\partial x^{(A)}_B},\bar{T}^{A \to B}_{*} \frac{\partial}{\partial p^{(A)}_B}) &=  1 \\
         \Omega^{(B)}(\bar{T}^{A \to B}_{*} \frac{\partial}{\partial x^{(A)}_{C_i}}, \bar{T}^{A \to B}_{*} \frac{\partial}{\partial p^{(A)}_{C_j}}) &=  \delta_{i,j} \\
         \Omega^{(B)}(\bar{T}^{A \to B}_{*} \frac{\partial}{\partial x^{(A)}_B},\bar{T}^{A \to B}_{*} \frac{\partial}{\partial p^{(A)}_{C_i}}) &=  0 \\
        \Omega^{(B)}(\bar{T}^{A \to B}_{*} \frac{\partial}{\partial x^{(A)}_{C_i}},\bar{T}^{A \to B}_{*} \frac{\partial}{\partial p^{(A)}_{B}}) &=  0.
    \end{align}
    From the first equality it follows that $\smash{c^{(B)}_A + \sum_i c^{(B)}_{C_i}= -1}$. From the second, $\smash{c^{(C_j)}_{C_i} = \delta_{i,j}}$. Finally, from the third, $\smash{c^{(C_i)}_{A} + \sum_j c^{(C_i)}_{C_j}=0}$ and from the fourth $\smash{c^{(B)}_{C_i} =0}$. Combined, this leads to the following system of partial differential equations
    \begin{align}
        \frac{\partial}{\partial p^{(A)}_B} p^{(B)}_{A}\left(\bar{T}^{A \to B}\left(\begin{psmallmatrix}
    x^{(A)}_B\\
    p^{(A)}_B
    \end{psmallmatrix}, \begin{psmallmatrix}
    x^{(A)}_{C_1}\\
    p^{(A)}_{C_1}
    \end{psmallmatrix}, \dots\right)\right) &= -1 \\
        \frac{\partial}{\partial p^{(A)}_B} p^{(B)}_{C_i}\left(\bar{T}^{A \to B}\left(\begin{psmallmatrix}
    x^{(A)}_B\\
    p^{(A)}_B
    \end{psmallmatrix}, \begin{psmallmatrix}
    x^{(A)}_{C_1}\\
    p^{(A)}_{C_1}
    \end{psmallmatrix}, \dots\right)\right) &= 0 \\
        \frac{\partial}{\partial p^{(A)}_{C_i}} p^{(B)}_{A}\left(\bar{T}^{A \to B}\left(\begin{psmallmatrix}
    x^{(A)}_B\\
    p^{(A)}_B
    \end{psmallmatrix}, \begin{psmallmatrix}
    x^{(A)}_{C_1}\\
    p^{(A)}_{C_1}
    \end{psmallmatrix}, \dots\right)\right) &= -1 \\
        \frac{\partial}{\partial p^{(A)}_{C_i}} p^{(B)}_{C_j}\left(\bar{T}^{A \to B}\left(\begin{psmallmatrix}
    x^{(A)}_B\\
    p^{(A)}_B
    \end{psmallmatrix}, \begin{psmallmatrix}
    x^{(A)}_{C_1}\\
    p^{(A)}_{C_1}
    \end{psmallmatrix}, \dots\right)\right) &= \delta_{i,j}.
    \end{align}
    Integrating and imposing the initial conditions given in \cref{pos:pos_mom_sep}, the claimed form of the transformation follows.
\end{proof}

\subsection{Robust version of the classical Paradox of the Third Particle}\label{app:robust} 

We now show that there is a version of the classical paradox we presented in the main text that is robust under small deviations. Suppose the distribution of the particle $B$ and the random bit $M$ from $A$'s perspective is 
\begin{equation}
    P^{(A)}_{BM}\left(\begin{psmallmatrix}
    x_{B}\\
    p_{B}
    \end{psmallmatrix},m\right) =  \frac{1}{2}P_B(x_B,p_B+km)
\end{equation}
for some $k$ and distribution $P_B$ on the phase space of particle $B$. Then the bit $M$ can, from $A$'s perspective, be perfectly retrieved from $B$ as long as $P_B(\cdot,\cdot)$ has disjoint support from $P_B(\cdot,\cdot+k)$. As the reference frame transformation is symplectic, from $B$'s perspective, the bit $M$ can be perfectly retrieved from $A$.

Suppose, from $A$'s perspective, there is a third particle $C$ such that the joint distribution over $BCM$ is
\begin{equation}\label{eq:pers_A_with_C}
    P^{(A)}_{BCM}\left(\begin{psmallmatrix}
    x_{B}\\
    p_{B}
    \end{psmallmatrix},\begin{psmallmatrix}
    x_{C}\\
    p_{C}
    \end{psmallmatrix},m\right) =  \frac{1}{2}P_B(x_B,p_B+km)P_C(x_C,p_C). 
\end{equation}
Then the following theorem bounds the success probability, from $B$'s perspective, of retrieving $M$ from $A$. 

\begin{resultbox}
    \begin{theorem}\label{thm:robustclassical}
        Let $k \in \mathbb{R}$ and $\smash{P^{(A)}_{BCM}}$ as in \cref{eq:pers_A_with_C} be the description of particles~$BC$ from $A$'s perspective. Then, from $B$'s perspective, the probability for successfully decoding the message $M$ from the system $A$ is bounded by
        \begin{equation}
            P^{(B)}_{\mathrm{succ}} \leq \frac{1}{4}\|Q_C - S_{k}Q_C\|_1 + \frac{1}{2}
        \end{equation}
        where $Q_C(p_C) \coloneq \int dx_C P_C(x_C,p_C)$ and $S_{k} Q_C(p) \coloneq Q_C(p+k)$.
    \end{theorem}
\end{resultbox}

Note that $\|Q_C - S_{k}Q_C\|_1$ can be arbitrarily small. For example, assume that we have $P_C(x_C,p_C) = \frac{1}{2R} \bar{P}_C(x_C)\chi_{[-R,R]}(p_C)$, where $\bar{P}_C(x_C)$ is some probability distribution. Then $Q_C = \frac{1}{2R}\chi_{[-R,R]}$, and we find, for $|k| \leq 2R$, that 
\begin{equation}
    \|Q_C - S_{k}Q_C\|_1 = \frac{1}{2R}\int dp \, |\chi_{[-R,R]}(p)-\chi_{[-R,R]}(p+k)| = \frac{|k|}{R}.
\end{equation}

\begin{proof}[Proof of \cref{thm:robustclassical}]
    The transformation of \cref{thm:classicalRF} yields $B$'s perspective
    \begin{equation}
        P^{(B)}_{ACM}\left(\begin{psmallmatrix}
            x_{A}\\
            p_{A}
        \end{psmallmatrix},\begin{psmallmatrix}
            x_{C}\\
            p_{C}
        \end{psmallmatrix},m\right) = \frac{1}{2}P_B(-x_A, -p_A-p_C+km) P_C(x_C-x_A,p_C). 
    \end{equation}
    To find how well $M$ can be recovered from $A$, we compute the marginal over the system $C$
    \begin{align}
        P^{(B)}_{AM}\left(\begin{psmallmatrix}
            x_{A}\\
            p_{A}
        \end{psmallmatrix},m\right) &=  \frac{1}{2} \int dp_C  P_B(-x_A, -p_A-p_C+km) \int dx_C P_C(x_C-x_A,p_C) \\
        &=  \frac{1}{2} \int dp_C  P_B(-x_A, -p_A-p_C+km) \int dx_C P_C(x_C,p_C) \\
        &=  \frac{1}{2} \int dp_C  P_B(-x_A, p_C) Q_C(-p_A-p_C+km) \\
        &= \frac{1}{2}  (P_B(-x_A, \cdot) * S_{km} Q_C)(-p_A)
    \end{align}
    where $S_{km} Q_C(p) = Q_C(p+km)$ and $Q_C(p_C) \coloneq \int dx_C P_C(x_C,p_C)$.  The maximum probability to be able to guess $m$ correctly when having access to $A$ is given by 
    \begin{align}
        P^{(B)}_{\mathrm{succ}} &= \frac{1}{2} + \frac{1}{4}\|P^{(B)}_{A}(x_A,p_A|m = 0) - P^{(B)}_{A}(x_A,p_A|m = 1)\|_1.
    \end{align}  
    We now bound the second term:
    \begin{align}
        \|(P_B * Q_C) - (P_B* S_{k} Q_C) \|_1 &= 2\sup_{0\leq\sigma \leq 1} \int dp_A dx_A  \sigma(p_A,x_A) \\ \nonumber
        &\quad \int dp_C  P_B(-x_A, p_C) (Q_C(-p_C-p_A) - Q_C(-p_C-p_A+k)) \\
        &= 2\sup_{0\leq\sigma \leq 1} \int dp_A dx_A \int dp_C \sigma(-p_A-p_C,x_A) \\ \nonumber
        &\quad P_B(-x_A, p_C) (Q_C(p_A) - Q_C(p_A+k)) \\
        &\leq 2\sup_{0\leq \bar{\sigma} \leq 1} \int dp_A \bar{\sigma}(p_A) (Q_C(p_A) - Q_C(p_A+k)) \\
        &= \|Q_C - S_{k}Q_C\|_1.
    \end{align}
    Thus, we find that $P^{(B)}_{\mathrm{succ}} \leq \frac{1}{2} + \frac{1}{4}\|Q_C - S_{k}Q_C\|_1$.
\end{proof}

\subsection{No-go theorem}
\begin{resultbox}
    \nogo*
\end{resultbox}

\begin{proof}
    First note that, when information preservation holds, for any two perspectives~$p,q$ that are maximally overlapping $T^{p \to q} = (T^{q \to p})^{-1}$. To see this, consider the map $T = T^{p \to q} \circ T^{q \to p}$, which is also invertible, as it is composed of invertible maps, per the information preservation assumption. Then, as the atlas is consistent, we have that $T = T \circ T$, which implies that $T$ is the identity. 

    We denote the perspectives in the graph of \cref{fig:no_go_framework} as follows: $p_1 = \{B\}_{\mathcal{D}_1}^{(A)}$, $p_2 = \{B,C\}_{\mathcal{D}_2}^{(A)}$, $q_1 = \{A\}_{\bar{\mathcal{D}}_1}^{(B)}$, and $q_2 = \{A,C\}_{\bar{\mathcal{D}}_2}^{(B)}$. We now prove the classical and quantum versions separately. 

    \textbf{Classical version.}
    As the graph of \cref{fig:no_go_framework} is a subgraph of the atlas, we focus on this graph. 
    Let the subsystem partitioning on the perspective $p_2$ be given by $\smash{\mathcal{P}_{2} = \mathcal{P}^{(A)}_{B} \times \mathcal{P}^{(A)}_C}$ and for $q_2$ given by $\smash{\mathcal{Q}_{2} = \mathcal{P}^{(B)}_{A} \times \mathcal{P}^{(B)}_C}$. Suppose that information preservation holds. For simplicity, we use the same symbols, like $T^{p_1 \to q_1}$, for the underlying symplectomorphisms on phase space rather than for the induced transformations on distributions. Note that the above proof of $(T^{p \to q})^{-1} = T^{q \to p}$ also applies to the underlying symplectomorphisms. By decomposability, the removal of $C$ is given by the trace, which on phase-space points, acts as the canonical projection. Consistency of the atlas therefore implies that
    \begin{equation}
        T^{p_1 \to q_1} \circ \Pi^{(A)}_{B} = \Pi^{(B)}_{A} \circ T^{p_2 \to q_2},
    \end{equation}
    where $\Pi^{(A)}_{B}$ and $\Pi^{(B)}_{A}$ are the canonical projection maps and $T^{p_1 \to q_1}$ and $T^{p_2 \to q_2}$ the reference frame transformations.
    Thus, for any $(t,s) \in \smash{\mathcal{P}^{(A)}_B \times \mathcal{P}^{(A)}_C}$
    \begin{equation}
        (T^{q_1 \to p_1} \times \id_C) \circ T^{p_2 \to q_2}(t,s) = (t,f(t,s))
    \end{equation}
    for some function $f$.
    Information preservation implies that the function $F \coloneq (T^{q_1 \to p_1} \times \id_C) \circ T^{p_2 \to q_2}$ is symplectic. Consider a pair $(t,s) \in \mathcal{P}^{(A)}_B \times \mathcal{P}^{(A)}_C$. Let $U \subset \mathcal{P}^{(A)}_{B}$ be an open set such that $t \in U$ and such that there is a chart $\psi$
    \begin{align}
        \psi: U &\to \mathbb{R}^{2m} \\
            t &\mapsto (x^{(A)}_1,p^{(A)}_1, \dots),
    \end{align}
    where the symplectic form on $U$ is given by $\sum_i dx^{(A)}_i \wedge dp^{(A)}_i$. Analogously, let $V^{(O)} \subset \mathcal{P}^{(O)}_C$ for $O \in \{A,B\}$ be an open set such that $s \in V^{(A)}$, $f(t,s) \in V^{(B)}$ and 
    \begin{align}
        \phi^{(O)}: V^{(O)} &\to \mathbb{R}^{2n} \\
            u &\mapsto (q^{(O)}_1,\pi^{(O)}_1, \dots)
    \end{align}
    is a chart such that the symplectic form on $V^{(O)}$ is given by $\sum^{n}_{i = 1} dq^{(O)}_i \wedge d\pi^{(O)}_i$. These open sets and charts exist due to Darboux's theorem. 
    Then, the fact that the function $F$ is symplectic implies that $\{F_* \frac{\partial}{\partial q^{(A)}_i}, F_* \frac{\partial}{\partial \pi^{(A)}_i}\}_{i}$ is a basis of $T_{(f(t,s))}\mathcal{P}^{(B)}_C$ such that 
    \begin{equation}
        \Omega(F_{*} \frac{\partial}{\partial q^{(A)}_i}, F_{*} \frac{\partial}{\partial \pi^{(A)}_j}) = \delta_{i,j}.
    \end{equation}
    Thus, there exist coefficients $a^{i}_k,b^{i}_\ell,c^{i}_k,d^{i}_{\ell}$ such that
    \begin{align}
        F_* \frac{\partial}{\partial x^{(A)}_i} &=  \frac{\partial}{\partial x^{(A)}_i} + \sum_k a^{i}_k F_*\frac{\partial}{\partial q^{(A)}_k} + \sum_\ell b^{i}_\ell F_*\frac{\partial}{\partial \pi^{(A)}_\ell} \\
        F_* \frac{\partial}{\partial p^{(A)}_i} &=  \frac{\partial}{\partial p^{(A)}_i} + \sum_k c^{i}_k F_*\frac{\partial}{\partial q^{(A)}_k} + \sum_\ell d^{i}_\ell F_*\frac{\partial}{\partial \pi^{(A)}_\ell}.
    \end{align}
    The requirement that $F$ is symplectic implies that 
    \begin{align}
        \Omega(F_* \frac{\partial}{\partial x^{(A)}_i}, F_* \frac{\partial}{\partial q^{(A)}_k}) &= - b^{i}_k = 0\\
        \Omega(F_* \frac{\partial}{\partial x^{(A)}_i}, F_* \frac{\partial}{\partial \pi^{(A)}_k}) &= a^{i}_k = 0 \\
        \Omega(F_* \frac{\partial}{\partial p^{(A)}_i}, F_* \frac{\partial}{\partial q^{(A)}_k}) &= -d^{i}_k = 0 \\
        \Omega(F_* \frac{\partial}{\partial p^{(A)}_i}, F_* \frac{\partial}{\partial \pi^{(A)}_k}) &= c^{i}_k= 0.
    \end{align}
    Therefore, for any differentiable path $g$ in $U$, it holds that $\frac{d}{dz}f(g(z), s)|_{z=0} = 0$. Because all phase spaces are connected, this means that $f$ is constant in its first entry. We write~$f(s)$ instead of $f(t,s)$. This then implies that 
    \begin{equation}
        T^{p_2 \to q_2}(t,s) = (T^{p_1 \to q_1}(t),f(s))
    \end{equation}
    which contradicts the assumption that $T^{p_2 \to q_2}$ is relational. As our assignment of subsystem partitioning was arbitrary, this conclusion holds for any such assignment.

    \textbf{Quantum version.}  As the graph of \cref{fig:no_go_framework} is a subgraph of the atlas, we focus on this graph. Let the subsystem partitioning on the perspective $p_2$ be given by $\smash{\mathcal{P}_{2} = \mathcal{H}^{(A)}_{B} \otimes \mathcal{H}^{(A)}_C}$ and for $q_2$ given by $\smash{\mathcal{Q}_{2} = \mathcal{H}^{(B)}_{A} \otimes \mathcal{H}^{(B)}_C}$. Suppose that information preservation holds. For simplicity, we use the same symbols, like $T^{p_1 \to q_1}$, for the underlying unitaries rather than for the induced transformations on density operators. Note that the above proof of $(T^{p \to q})^{-1} = T^{q \to p}$ also applies to the underlying unitaries, up to a global phase, which is irrelevant in what follows.

    Consider a state $\ket{\phi} \in \mathcal{H}^{(A)}_B$. By decomposability, the removal of $C$ is given by the partial trace, so the consistency of the atlas implies that 
    \begin{equation}
        T^{q_1 \to p_1} \tr_C\!\left(T^{p_2 \to q_2}(\ketbra{\phi}_B \otimes \rho_C)(T^{p_2 \to q_2})^{\dagger}\right)(T^{q_1 \to p_1})^{\dagger} = \ketbra{\phi}_B.
    \end{equation}
    We define the unitary $W \coloneq (T^{q_1 \to p_1}\otimes \mathbbm{1}_C)\, T^{p_2 \to q_2}$. The above equation states that the marginal of $W(\ketbra{\phi}_B \otimes \rho_C)W^{\dagger}$ on the first subsystem is the pure state $\ketbra{\phi}_B$. Because a state that is pure on $B$ cannot be correlated with another subsystem, we have that
    \begin{equation}
        W(\ketbra{\phi}_B \otimes \rho_C)W^{\dagger} = \ketbra{\phi}_B \otimes \mathcal{N}(\ketbra{\phi}_B \otimes \rho_C),
    \end{equation}
    for some map $\mathcal{N}$. In particular, for a pure state $\rho_C = \ketbra{\psi}_C$, the left-hand side is pure, and therefore 
    \begin{equation}
        W(\ket{\phi}_B \otimes \ket{\psi}_C) = \ket{\phi}_B \otimes V(\ket{\phi}_B, \ket{\psi}_C),
    \end{equation}
    where $V(\ket{\phi}_B, \ket{\psi}_C) = {}_B\!\bra{\phi}W(\ket{\phi}_B \otimes \ket{\psi}_C)$. Let now $\ket{0}$ and $\ket{1}$ be two orthogonal unit vectors, then from the linearity of $W$ it follows that 
    \begin{equation}
       \frac{1}{\sqrt{2}}(\ket{0}_B \otimes V(\ket{0}_B, \ket{\psi}_C) + \ket{1}_B \otimes V(\ket{1}_B, \ket{\psi}_C)) = \frac{1}{\sqrt{2}}(\ket{0}_B+\ket{1}_B) \otimes V(\frac{1}{\sqrt{2}}(\ket{0}_B+\ket{1}_B), \ket{\psi}_C).
    \end{equation}
    Therefore, we have that 
    \begin{equation}
        V(\ket{0}_B, \ket{\psi}_C)  = V(\ket{1}_B, \ket{\psi}_C) = V(\frac{1}{\sqrt{2}}(\ket{0}_B+\ket{1}_B), \ket{\psi}_C).
    \end{equation}
    Let $\{\ket{e_i}_B\}_{i}$ be an orthonormal basis of the $B$ system, then because the above holds for all pairs of orthogonal vectors, we have that 
    \begin{equation}
       \forall i,j: V(\ket{e_i}_B, \ket{\psi}_C)  = V(\ket{e_j}_B, \ket{\psi}_C).
    \end{equation}
    Let now $\ket{\phi} = \sum_{i = 1}^{N} a_i \ket{e_i}$ be a state, then using this result and the linearity of $W$, we have that
    \begin{align}
        W(\ket{\phi}_B \otimes \ket{\psi}_C) &= \ket{\phi}_B \otimes V(\ket{\phi}_B, \ket{\psi}_C) \\
        &= \sum_i a_i \ket{e_i}_B\otimes V(\ket{e_i}_B, \ket{\psi} ) \\
        &= \ket{\phi}_B \otimes V(\ket{e_1}_B, \ket{\psi} ).
    \end{align}
    Thus, $V(\ket{\phi}_B, \ket{\psi}_C)$ is independent of $\ket{\phi}_B$. Since vectors that are finite linear combinations of the basis vectors are dense and $W$ is bounded, this holds for all pure states of~$B$.
    
    Consequently, $W = \mathbbm{1}_B \otimes V$ for some unitary $V$, and hence $T^{p_2 \to q_2} = (T^{p_1 \to q_1} \otimes \mathbbm{1}_C)(\mathbbm{1}_B \otimes V) = T^{p_1 \to q_1} \otimes V$. The reference frame transformation factorizes, which contradicts the relationality assumption. As the subsystem assignment was general, the same argument holds for any subsystem assignment.
\end{proof}

\begin{resultbox}
    \subsysrelativity*
\end{resultbox}

Before we prove this result, the following lemma is useful. 
\begin{resultbox}
    \begin{lemma}\label{lem:algebra}
        Consider a product phase space $\mathcal{P}_{B} \times \mathcal{P}_C$ and denote by $\mathcal{A}_B,\mathcal{A}_C$ the algebras of observables on systems $B$ and $C$, then, for any observable $f$ on $\mathcal{P}_B \times \mathcal{P}_C$,
        \begin{equation}
            \{f,\mathcal{A}_C\} = 0 \implies f \in \mathcal{A}_B.
        \end{equation}
    \end{lemma}
\end{resultbox}
\begin{proof}
    Let $(z_B,z_C)$ be a point in $\mathcal{P}_{B} \times \mathcal{P}_C$ and consider a chart $\psi: U_B \times U_C \to \mathbb{R}^{2n}$ where $(z_B,z_C) \in U_B \times U_C$, that respects the subsystem partition $(x_{B,1}, p_{B,1}, \dots, x_{C,1}, p_{C,1}, \dots)$ and such that $\Omega$ has standard form on $U_B \times U_C$. Consider the coordinate function $p_{C,i}$. We multiply the coordinate function $p_{C,i}$ by a smooth function $\beta$ on $\mathcal{P}_C$ such that $\beta = 1$ in a neighbourhood $V_C$ around $z_C$ and that is supported only on $U_C$. The function $\beta p_{C,i}$ is then an element of $\mathcal{A}_C$. Then, we find that within $U_B \times V_C$
    \begin{equation}
        0 = \{f, \beta p_{C,i}\} = \frac{\partial f}{\partial x_{C,i}}.
    \end{equation}
    As this holds for all $x_{C,i}$ and, by an analogous argument, also for all $p_{C,i}$ using that the phase space $\mathcal{P}_C$ is connected, we find that $f$ is constant in all coordinates of the~$C$ system. Therefore, $f \in \mathcal{A}_B$.
\end{proof}

\begin{proof}[Proof of \cref{thm:subsysrelativity}]
    We separately prove the classical and the quantum version. Because information preservation holds, we denote by $\bar{T}^{A \to B}$ the function on phase space in the classical case or the unitary in the quantum case instead of the induced transformation on the state space. Moreover, we use the notation $\bar{T}^{B \to A} = (\bar{T}^{A \to B})^{-1}$.

    \textbf{Classical version.} Suppose that relationality does not hold. Then the reference frame transformation factorizes, and it follows directly that subsystems are not relative.

    Before proving the other direction, we introduce some notation. We define 
    \begin{equation}
        (\bar{T}^{A \to B})_{*} f \coloneq f \circ (\bar{T}^{A \to B})^{-1}.
    \end{equation}
    
    Suppose that relativity of subsystems is not true, i.e., $(\bar{T}^{A \to B})_{*} \mathcal{A}_{C}^{(A)} = \mathcal{A}_{C}^{(B)}$. Then, $0 = \smash{\{\mathcal{A}_{B}^{(A)}, (\bar{T}^{B \to A})_{*} \mathcal{A}_{C}^{(B)} \}}$, which implies that $\smash{\{(\bar{T}^{A \to B})_{*} \mathcal{A}_{B}^{(A)},  \mathcal{A}_{C}^{(B)} \}} = 0$.
    Therefore, by \cref{lem:algebra}, it follows that $\smash{(\bar{T}^{A \to B})_{*} \mathcal{A}_{B}^{(A)} \subset  \mathcal{A}_{A}^{(B)}}$. Applying a similar argument with the roles of $A$ and $B$ exchanged yields
    \begin{equation}\label{eq:alg_preservation}
        (\bar{T}^{A \to B})_{*} \mathcal{A}_{B}^{(A)} =  \mathcal{A}_{A}^{(B)}.
    \end{equation}
    
    We define the function $\smash{R  =  \Pi_B^{(A)}\circ \bar{T}^{B \to A}}$, where $\smash{\Pi_B^{(A)}}$ is the canonical projection of $\smash{\mathcal{P}_B^{(A)}\times\mathcal{P}_C^{(A)}}$ onto the $B$ factor. Analogously, we define another function $\smash{R_C  = \Pi_C^{(A)} \circ \bar{T}^{B \to A}}$, where~$\smash{\Pi_C^{(A)}}$ is the canonical projection of $\mathcal{P}_B^{(A)}\times\mathcal{P}_C^{(A)}$ onto the $C$ factor. 

    Let $\smash{\psi^{(A)}_B, \psi^{(A)}_C}$ be charts for $\smash{\mathcal{P}_B^{(A)}, \mathcal{P}_C^{(A)}}$, respectively. Denote by $\smash{\psi^{(A)}_{B,i}}$ the $i$th coordinate, and analogously for the other chart. Let $\smash{\beta_B^{(A)}, \beta_C^{(A)}}$ be smooth functions on $\mathcal{P}_B^{(A)}$ and $\mathcal{P}_C^{(A)}$, respectively, each supported only on the domain of the corresponding chart and equal to~$1$ on an open subset of that domain.
     
    We multiply the coordinate function, $\psi^{(A)}_{B,i}$, which is only defined on the chart, by $\beta_B^{(A)}$  to construct a function $\smash{\beta_B^{(A)} (\psi_{B,i}^{(A)} \circ \Pi_B^{(A)}) \in \mathcal{A}_B^{(A)}}$. By \cref{eq:alg_preservation} its image under $(\bar{T}^{A \to B})_{*}$ lies in $\smash{\mathcal{A}_A^{(B)}}$. Hence, on this chart, $R$ is constant in its second entry. Because this holds for any chart, and the manifold $\mathcal{P}_C^{(B)}$ is connected, $R$ is globally constant in its second entry. Analogously, we multiply the coordinate functions of system $C$ by $\beta_C^{(A)}$ to construct a function $\beta_C^{(A)}(\smash{\psi_{C,i}^{(A)} \circ \Pi_C^{(A)}}) \in \mathcal{A}_C^{(A)}$. The assumption that subsystem relativity is not true implies that its image under $(\bar{T}^{A \to B})_{*}$ lies in $\smash{\mathcal{A}_C^{(B)}}$. Thus, $R_C$ is constant in its first entry on this chart. As $\mathcal{P}_A^{(B)}$ is connected and the chart is arbitrary, $R_C$ is globally constant in its first entry.
   
   Therefore, the two functions define maps $R: \mathcal{P}_A^{(B)} \to \mathcal{P}_B^{(A)}$ and $R_C: \mathcal{P}_C^{(B)} \to \mathcal{P}_C^{(A)}$. Hence, $\smash{\bar{T}^{B \to A} = R \times R_C}$, i.e., $\smash{\bar{T}^{B \to A}}$ factorizes. Thus, we have shown that relationality does not hold. 

    \textbf{Quantum version.} Suppose that relationality does not hold. Then the reference frame transformation factorizes, and it follows directly that subsystems are not relative. 

    Suppose that relativity of subsystems is false, then $\bar{T}^{A \to B} \mathcal{A}_{C}^{(A)} (\bar{T}^{A \to B})^{\dagger} = \mathcal{A}_{C}^{(B)}$. This condition implies that 
    \begin{equation}\label{eq:no_influence_B}
        0 = [\bar{T}^{B \to A} \mathcal{A}_C^{(B)} (\bar{T}^{B \to A})^{\dagger}, \bar{T}^{B \to A}\mathcal{A}_A^{(B)}(\bar{T}^{B \to A})^{\dagger}] = [\mathcal{A}_C^{(A)}, \bar{T}^{B \to A}\mathcal{A}_A^{(B)}(\bar{T}^{B \to A})^{\dagger}].
    \end{equation}
    As $\mathcal{A}_B^{(A)}$ is the commutant of $\mathcal{A}_C^{(A)}$, it follows from \cref{eq:no_influence_B} that $\bar{T}^{B \to A}\mathcal{A}_A^{(B)}(\bar{T}^{B \to A})^{\dagger} \subset \mathcal{A}_B^{(A)}$. Analogously, we find that 
    \begin{equation}\label{eq:no_influence_A}
        0 = [\bar{T}^{A \to B} \mathcal{A}_C^{(A)} (\bar{T}^{A \to B})^{\dagger}, \bar{T}^{A \to B}\mathcal{A}_B^{(A)}(\bar{T}^{A \to B})^{\dagger}] = [\mathcal{A}_C^{(B)}, \bar{T}^{A \to B}\mathcal{A}_B^{(A)}(\bar{T}^{A \to B})^{\dagger}].
    \end{equation}
    As $\mathcal{A}_A^{(B)}$ is the commutant of $\mathcal{A}_C^{(B)}$, it follows from \cref{eq:no_influence_A} that $\bar{T}^{A \to B}\mathcal{A}_B^{(A)}(\bar{T}^{A \to B})^{\dagger} \subset \mathcal{A}_A^{(B)}$. Combined with the inclusion above, it follows that $\bar{T}^{B \to A}\mathcal{A}_A^{(B)}(\bar{T}^{B \to A})^{\dagger} = \mathcal{A}_B^{(A)}$. 

    Let $\ket{0}_A \in \mathcal{H}_A^{(B)}$ and $\ket{0}_C \in \mathcal{H}_C^{(B)}$ be two normalized vectors, $\{\ket{i}_B\}$ an orthonormal basis of $\mathcal{H}_B^{(A)}$ and $\{\ket{k}_C\}$ of $\mathcal{H}_C^{(A)}$. Let $\ket{i}_B,\ket{\alpha}_{B}$ be two elements of $\{\ket{i}_B\}$ and $\ket{k}_{C},\ket{\gamma}_{C}$ two elements of $\{\ket{k}_C\}$.
    
    Since $\ketbra{i}{\alpha}_{B}\otimes \mathbbm{1}_C \in \mathcal{A}_B^{(A)}$ and $\mathbbm{1}_B \otimes \ketbra{k}{\gamma}_{C} \in \mathcal{A}_C^{(A)}$, the assumption that subsystems are not relative and $\bar{T}^{B \to A}\mathcal{A}_A^{(B)}(\bar{T}^{B \to A})^{\dagger} = \mathcal{A}_B^{(A)}$ imply that there are bounded operators $X_A^{i\alpha}$ on $\mathcal{H}_A^{(B)}$ and $Y_C^{k\gamma}$ on $\mathcal{H}_C^{(B)}$ such that
    \begin{align}
        (\bar{T}^{B \to A})^{\dagger}(\ketbra{i}{\alpha}_{B}\otimes \mathbbm{1}_C) \bar{T}^{B \to A} &= X_A^{i\alpha} \otimes \mathbbm{1}_C \label{eq:X_operator}\\
        (\bar{T}^{B \to A})^{\dagger} (\mathbbm{1}_B \otimes \ketbra{k}{\gamma}_{C}) \bar{T}^{B \to A} &= \mathbbm{1}_A \otimes Y_C^{k\gamma}. \label{eq:Y_operator}
    \end{align}
    We define the following two quantum channels
    \begin{align}
        \mathcal{E}(\rho_A^{(B)}) &= \tr_C(\bar{T}^{B \to A}(\rho_A^{(B)} \otimes \ketbra{0}_C) (\bar{T}^{B \to A})^{\dagger}) \\
        \mathcal{E}_C(\rho_C^{(B)}) &= \tr_B(\bar{T}^{B \to A}(\ketbra{0}_A \otimes \rho_C^{(B)}) (\bar{T}^{B \to A})^{\dagger}).
    \end{align}
    We now show that $\bar{T}^{B \to A}(\rho^{(B)}_{AC})(\bar{T}^{B \to A})^{\dagger} =  (\mathcal{E}\otimes \mathcal{E}_C)(\rho_{AC}^{(B)})$. Let $\{\ket{a}_A\}$ be an orthonormal basis of $\mathcal{H}_A^{(B)}$ and $\{\ket{c}_C\}$ of $\mathcal{H}_C^{(B)}$. First, we show that for operators of the form $\ketbra{a}{b}_A \otimes \ketbra{c}{d}_C$ the action of these two channels is equal, by comparing their matrix elements with respect to the basis $\{\ket{i}_B\}$ and $\{\ket{k}_C\}$ introduced above
    \begin{align}
        \bra{\alpha}\mathcal{E}(\ketbra{a}{b}_A)\ket{i}_B &= \tr((\bar{T}^{B \to A})^{\dagger}(\ketbra{i}{\alpha}_{B}\otimes \mathbbm{1}_C) \bar{T}^{B \to A} (\ketbra{a}{b}_A \otimes \ketbra{0}_C)) \\
        &= \tr(X_A^{i\alpha} \ketbra{a}{b}_A) \\
        \bra{\gamma}\mathcal{E}_C(\ketbra{c}{d}_C)\ket{k}_C &= \tr((\bar{T}^{B \to A})^{\dagger} (\mathbbm{1}_B \otimes \ketbra{k}{\gamma}_{C}) \bar{T}^{B \to A} (\ketbra{0}_A \otimes \ketbra{c}{d}_C)) \\
        &= \tr(Y_C^{k\gamma} \ketbra{c}{d}_C).
    \end{align}
    Let us now calculate the matrix elements of $\bar{T}^{B \to A}(\ketbra{a}{b}_A \otimes \ketbra{c}{d}_C)(\bar{T}^{B \to A})^{\dagger}$:
    \begin{align}
       &\bra{\alpha,\gamma}_{BC} \bar{T}^{B \to A}(\ketbra{a}{b}_A \otimes \ketbra{c}{d}_C) (\bar{T}^{B \to A})^{\dagger}\ket{i,k}_{BC} \\
       &= \tr((\bar{T}^{B \to A})^{\dagger}(\ketbra{i}{\alpha}_{B}\otimes \mathbbm{1}_C) \bar{T}^{B \to A} (\bar{T}^{B \to A})^{\dagger} (\mathbbm{1}_B \otimes \ketbra{k}{\gamma}_{C}) \bar{T}^{B \to A} (\ketbra{a}{b}_A \otimes \ketbra{c}{d}_C)) \\
        &= \tr((X_A^{i\alpha} \otimes Y_C^{k\gamma}) (\ketbra{a}{b}_A \otimes \ketbra{c}{d}_C)) \\
        &= \tr(X_A^{i\alpha} \ketbra{a}{b}_A) \tr(Y_C^{k\gamma} \ketbra{c}{d}_C) \\
        &= \bra{\alpha,\gamma}_{BC} (\mathcal{E}\otimes \mathcal{E}_C)(\ketbra{a}{b}_A \otimes \ketbra{c}{d}_C) \ket{i,k}_{BC}.
    \end{align}

    From linearity and the fact that both are quantum channels, and, thus, continuous with respect to the trace norm, we can use that the span of product operators is dense within the set of trace class operators. Thus, 
    \begin{equation}
        \bar{T}^{B \to A}\rho_{AC}^{(B)} (\bar{T}^{B \to A})^{\dagger} = (\mathcal{E}\otimes \mathcal{E}_C)(\rho_{AC}^{(B)}),
    \end{equation}
    which means that $\bar{T}^{B \to A}$ factorizes and relationality does not hold.
\end{proof}

Note that the quantum version of \cref{thm:subsysrelativity} also follows from standard results in the quantum causal model literature. In the language of \cite{Ormrod_2023}, the condition in \cref{eq:no_influence_B} states that there is no influence from $A$ to $C$ through the transformation $\bar{T}^{B \to A}$. Furthermore, when subsystems are not relative it follows that 
\begin{equation}
    [\mathcal{A}_B^{(A)}, \bar{T}^{B\to A}\mathcal{A}_C^{(B)}(\bar{T}^{B\to A})^\dagger] = 0,
\end{equation}
which means that there is also no influence from $C$ to $B$ through the transformation $\bar{T}^{B \to A}$. It is a standard result that such a unitary transformation factorizes; see, for example, \cite{Lorenz_2021}.  
\end{document}

%% file: preamble.tex
\usepackage{amsmath}
\usepackage{amssymb}
\usepackage{amsthm}
\usepackage{physics}
\usepackage{mathtools}
\usepackage{thmtools} 
\usepackage{thm-restate}
\usepackage{bbm}
\usepackage{cancel}

\allowdisplaybreaks

\PassOptionsToPackage{hyphens}{url}
\usepackage{url}
\usepackage{hyperref}
\hypersetup{colorlinks=true, linkcolor=Black, citecolor=Blue, urlcolor=Blue,hypertexnames=false}
\usepackage[capitalize,noabbrev]{cleveref}
\usepackage{bookmark}
\usepackage{footnotehyper}

\usepackage{tabularx}
\usepackage{booktabs}
\usepackage{multicol} 
\usepackage{makecell}
\usepackage{array,multirow}
\usepackage{caption}
\usepackage{subcaption}
\usepackage{graphicx}

\usepackage[dvipsnames]{xcolor}
\usepackage{csquotes}
\usepackage{datetime} 
\usepackage{datetime2}
\usepackage{fancyhdr}
\usepackage{titlesec}
\usepackage{sectsty} 
\usepackage{listings}
\usepackage{enumitem}
\usepackage{authblk}
\usepackage{svg}
\usepackage[left=1.15in, right=1.15in, top=1.2in, bottom=1in, headsep=.25in]{geometry}
\usepackage[bottom]{footmisc}
\usepackage[english]{babel}

\usepackage{tikz}
\usepackage{tikz-cd}
\usepackage{pgfplots}
\usetikzlibrary{calc}
\pgfplotsset{compat=1.18}
\usepgfplotslibrary{fillbetween}
\usepackage{relsize}
\usetikzlibrary{decorations.fractals,positioning}
\usetikzlibrary{arrows.meta}
\usetikzlibrary{positioning, shapes.geometric, calc}
\usetikzlibrary{shadows.blur}
\usetikzlibrary{shapes.misc}
\usetikzlibrary{decorations.pathmorphing}
\usetikzlibrary{arrows.meta}
\usetikzlibrary{decorations.markings}

\allowdisplaybreaks

\sectionfont{\fontsize{13}{15}\selectfont}
\subsectionfont{\fontsize{11}{15}\selectfont}

\renewenvironment{abstract}
{\small\begin{quote}\noindent \par{\sc \abstractname.}}
{\noindent\end{quote}}

\usepackage[sorting=none,backend=biber,style=numeric-comp,firstinits=false,giveninits=false,maxbibnames=99,maxcitenames=99]{biblatex}
\DeclareFieldFormat{titlecase}{\MakeTitleCase{#1}}

\newrobustcmd{\MakeTitleCase}[1]{%
  \ifthenelse{\ifcurrentfield{booktitle}\OR\ifcurrentfield{booksubtitle}%
    \OR\ifcurrentfield{maintitle}\OR\ifcurrentfield{mainsubtitle}%
    \OR\ifcurrentfield{journaltitle}\OR\ifcurrentfield{journalsubtitle}%
    \OR\ifcurrentfield{issuetitle}\OR\ifcurrentfield{issuesubtitle}%
    \OR\ifentrytype{book}\OR\ifentrytype{mvbook}\OR\ifentrytype{bookinbook}%
    \OR\ifentrytype{booklet}\OR\ifentrytype{suppbook}%
    \OR\ifentrytype{collection}\OR\ifentrytype{mvcollection}%
    \OR\ifentrytype{suppcollection}\OR\ifentrytype{manual}%
    \OR\ifentrytype{periodical}\OR\ifentrytype{suppperiodical}%
    \OR\ifentrytype{proceedings}\OR\ifentrytype{mvproceedings}%
    \OR\ifentrytype{reference}\OR\ifentrytype{mvreference}%
    \OR\ifentrytype{report}\OR\ifentrytype{thesis}}
    {#1}
    {\MakeSentenceCase{#1}}}
\AtBeginBibliography{\emergencystretch=3em}

\usepackage[subfigure]{tocloft}
\newtheorem{theorem}{Theorem}
\newtheorem{definition}{Definition}
\newtheorem{lemma}[theorem]{Lemma}

\newtheorem{postulate}{Postulate}

\Crefname{postulate}{Postulate}{Postulates}
\crefname{postulate}{Postulate}{Postulates}

\newtheoremstyle{example}{}{}{}{}{\bfseries}{.}{.5em}{Example \thmnumber{#2}}
\theoremstyle{example}
\newtheorem{example}{Example}

\newtheoremstyle{named}{}{}{\itshape}{}{\bfseries}{.}{.5em}{\thmnote{#3}}
\theoremstyle{named}

\makeatletter
\def\namedlabel#1#2{\begingroup
   \def\@currentlabel{#2}%
   \label{#1}\endgroup
}
\makeatother

\crefname{assumption}{assumption}{assumptions}

\numberwithin{equation}{section}

\crefname{equation}{}{}
\Crefname{equation}{}{}

\usepackage{tcolorbox}
\tcbuselibrary{theorems}
\tcbuselibrary{skins}
\tcbuselibrary{breakable}
\tcbuselibrary{hooks}

\newtcolorbox{resultbox}{%
    enhanced,breakable,
    frame hidden,
    colframe=Blue,
    colback=Blue!3,
    coltitle=black,
    fonttitle=\bfseries,
    colbacktitle=Blue!3,
    borderline={0.2mm}{0mm}{Blue},
    top=1mm,
    bottom=1mm,
    left=1.5mm,
    right=1.5mm,
    before skip=2.5ex,
    after skip=2.5ex
}

\newtcolorbox{postulatebox}
{
    enhanced,breakable,frame hidden,
    colframe=BrickRed,
    colback=BrickRed!3,
    coltitle=black,
    fonttitle=\bfseries,
    colbacktitle=BrickRed!3,
    borderline={0.2mm}{0mm}{BrickRed},
    top=1mm,
    bottom=1mm,
    left=1.5mm,
    right=1.5mm,
    before skip=2ex,
    after skip=2ex
}

\makeatletter
\newtcbtheorem[auto counter]{resolution}{Resolution}
{
    theorem style=plain,enhanced,breakable,frame hidden,
    colframe=ForestGreen,
    colback=ForestGreen!3,
    coltitle=black,
    fonttitle=\bfseries,
    colbacktitle=ForestGreen!3,
    borderline={0.2mm}{0mm}{ForestGreen},
    top=1mm,
    bottom=1mm,
    left=1.5mm,
    right=1.5mm,
    before skip=2ex,
    after skip=2ex
}
{res}
\makeatother

\Crefname{tcb@cnt@resolution}{Resolution}{Resolutions}
\crefname{tcb@cnt@resolution}{Resolution}{Resolutions}

\makeatletter
\AddToHook{cmd/appendix/before}{\def\cref@section@alias{appendix}}
\AddToHook{cmd/appendix/before}{\def\cref@subsection@alias{appendix}}
\makeatother

\newcommand{\id}{\mathrm{id}}